\documentclass[showpacs,preprintnumbers,amsmath,amssymb]{revtex4}
\newtheorem{theorem}{Theorem}

\newtheorem{lemma}{Lemma}
\newtheorem{corollary}{Corollary}

\newcommand{\BlackBox}{\rule{1.5ex}{1.5ex}}  
\newenvironment{proof}{\par\noindent{\bf Proof:\
}}{\hfill\BlackBox\\[2mm]}

\usepackage{graphicx}
\usepackage{dcolumn}
\usepackage{bm}

\begin{document}
\title{
Relation between strength of interaction and accuracy of measurement 
for a quantum measurement
}

\author{Takayuki Miyadera}
\affiliation{%
Research Center for Information Security (RCIS), \\
National Institute of Advanced Industrial
Science and Technology (AIST). \\
Daibiru building 1003,
Sotokanda, Chiyoda-ku, Tokyo, 101-0021, Japan.
\\
(E-mail: miyadera-takayuki@aist.go.jp)
}%


\date{\today}

\begin{abstract}
The process of measuring a two-level quantum system 
was examined by applying Hamiltonian formalism.  
For the measurement of an observable 
that does not commute with the system Hamiltonian, 
a non-trivial relationship among the strength of 
interaction, the time interval of the process, and 
the accuracy of the measurement was 
obtained.  
Particularly, to achieve an error-free measurement of
 such an observable, 
 a condition stating that the interaction Hamiltonian 
 does not commute with the system Hamiltonian needs to 
 be satisfied. 
\end{abstract}
\maketitle
\section{Introduction}
Von Neumann formulated the measurement process as 
the dynamics of a compound system that comprises 
a system and an apparatus \cite{vN}. This theory is now widely accepted 
and has been extensively investigated by researchers. 
In the formalism, a measurement process is determined 
by identifying the following variables: 
the Hilbert space of an apparatus, 
the initial state of the apparatus, and 
the unitary operator acting on a composite system. 
A good measurement process is obtained by 
cleverly choosing these variables. 
Because the process must be realized 
by physical devices at least in principle, 
its dynamics is often identified by 
a Hamiltonian operator. 
Several previous models have identified the dynamics 
in this manner. 
They include measurement processes of 
various physical quantities and 
approximated joint measurements of 
noncommutative observables \cite{Buschbook}. 
In contrast, another 
direction of research 
investigates the limitations of this process. 
In this case, because the process does not 
have to be realized by realistic physical devices, 
its dynamics are often identified by giving 
possible unitary operators of the composite system. 
A typical result
in this direction takes the
 form of an impossibility theorem: 
some operations are not achievable 
even if one may use unitary operators. 
For instance, the uncertainty principle 
states the impossibility of jointly measuring  
noncommutative observables 
\cite{BuschUncertainty,Buscherrorbar,Werner,
Appleby,OzawaUncertain,Janssens,
MiyaHeisen}, 
and the 
Wigner-Araki-Yanase 
theorem states the impossibility of 
performing precise measurements in the presence of 
an additive conserved quantity 
\cite{Wigner,ArakiYanase,Ozawa,Beltrametti,MiyaWAY,Loveridge}. 
\par
In the present study, we examine a problem that 
takes an intermediate position between the above 
two research directions. 
We use a Hamiltonian to identify a 
measurement process and discuss its limitations.  
Let us specifically describe the problem. 
Assume that there exists a system whose 
dynamics is governed by a system Hamiltonian 
$H_S$. To measure an observable $Q$ of this system, 
one needs to prepare an apparatus and introduce 
an interaction between the system and the apparatus. 
This interaction is described by an interaction 
Hamiltonian $V$. 
We want to determine the strength of 
interaction $V$ and the time interval
$\tau$ required for measuring $Q$. 
We investigate the measurement of 
the simplest system, a two-level 
quantum system, and obtain a non-trivial relationship 
among the strength of interaction $V$, the time interval $\tau$, and 
the accuracy of the measurement 
of $Q$ that does not commute with $H_S$. 
The conclusion contains a simple interpretation of 
this result by using 
the uncertainty principle. 
\section{Formulation and Results}
\subsection{Formulation}
In the following section, we study the dynamics of 
quantum measurement by applying Hamiltonian formalism 
(See \cite{Buschbook} for a general treatment of measurement). 
The dynamics of quantum measurement is 
described by an interaction process between 
a system and an apparatus. 
Suppose that the system is 
described by a
Hilbert space ${\cal H}_S$, and the apparatus 
is described by a Hilbert space ${\cal H}_A$. 
The observable to be measured 
is denoted by a self-adjoint operator $Q$ 
on ${\cal H}_S$. 
$Q$ is diagonalized as $Q=\sum_{q\in S_Q} q P_q$, 
where $\{P_q\}_{q\in S_Q}$ forms a projection-valued measure (PVM). 
That is, each $P_q$ is a projection operator 
and $\sum_{q\in S_Q}  P_q ={\bf 1}_S$ holds.  
Measurement is a physical process that transfers 
the value of $Q$ at time $t=0$ to an observable of the apparatus 
at a certain time $t=\tau$. 
The entire state of the composite 
system evolves from time $t=0$ to $t=\tau$ 
following the Schr\"odinger (or von Neumann) equation. 
The total Hamiltonian is written as 
\begin{eqnarray*}
H = H_S \otimes {\bf 1}_A +{\bf 1}_S\otimes H_A +V, 
\end{eqnarray*}
where $H_S$ (resp. $H_A$) represents 
the Hamiltonian acting only on the system (resp. the apparatus), 
and $V$ represents the interaction Hamiltonian. 
At time $t=0$, the initial state has a product form such as:
\begin{eqnarray*}
\rho(0) =\rho_S\otimes |\Omega \rangle \langle \Omega |, 
\end{eqnarray*}
where the unit vector 
$|\Omega \rangle \in {\cal H}_A$ does not depend on 
$\rho_S$. 
At time $t=\tau$, the state of the composite system becomes 
$\rho(\tau)=U(\tau)\rho(0)U(\tau)^*$, where 
$U(\tau):=\exp(-i\frac{H \tau}{\hbar})$.  
Without the interaction term $V$, the state 
keeps its product form; therefore, no 
information transfer from the system to the apparatus 
occurs. In this study, we investigate 
how large $V$ and $\tau$ should be in order to describe 
a measurement process. 
\subsection{Measurement of a two-level system}
 Throughout this paper, 
 we assume 
 that the system is a two-level quantum system and 
 that the observable $Q$ has only 
 two outcomes: $1$ and $0$. 
 That is, $Q$ is a projection operator on ${\cal H}_S$.
 We write the eigenstates as 
 $|q_1\rangle$ and $|q_0\rangle$, 
 where $Q|q_1\rangle =|q_1\rangle $ 
 and $Q|q_0\rangle =0$ hold. 
 We consider two initial states of the system, 
 $\rho^S_0=|q_0\rangle \langle q_0|$ 
 and $\rho^S_1=|q_1\rangle \langle q_1|$.
From the viewpoint of information transfer, 
the quality of measurement is characterized by 
its ability to distinguish between the 
states of the apparatus at $t=\tau$. 
The time evolution of the two initial states 
results in the following two 
final states of the apparatus: 
\begin{eqnarray*}
\rho_0^A(\tau)&=&\mbox{tr}_{{\cal H}_S} 
(U(\tau) (\rho^S_0 \otimes |\Omega\rangle \langle 
\Omega |)U(\tau)^*) \\
\rho_1^A(\tau)&=&\mbox{tr}_{{\cal H}_S} 
(U(\tau) (\rho^S_1 \otimes |\Omega\rangle \langle 
\Omega |)U(\tau)^*),
\end{eqnarray*}
where $\mbox{tr}_{{\cal H}_S}$ represents 
a partial trace with respect to ${\cal H}_S$. 
The process can only have an error-free measurement when 
these two final states of the apparatus are 
perfectly distinguishable. 
Note that we do not impose any condition on the states of
the system after the measurement, 
while 
a repeatability condition is often imposed in literatures 
that discuss measurement. 
A measurement that satisfies 
the repeatability condition is a special kind of measurement 
called an ideal measurement. 
We employ a quantity called fidelity 
as a measure of the distinguishability 
between states of the apparatus after the interaction. 
The fidelity between states $\rho$ and $\sigma$ on 
a Hilbert space ${\cal H}$ is defined by 
\begin{eqnarray*}
F(\rho,\sigma):=\mbox{tr}(\sqrt{\sigma^{1/2}\rho\sigma^{1/2}}),
\end{eqnarray*}
which is symmetric with 
respect to $\rho$ and $\sigma$ and 
satisfies the condition $0\leq F(\rho,\sigma)\leq 1$. 
The following theorem is useful for understanding its operational 
meaning. 
\begin{lemma}\label{dare}\cite{Fidelity}
Suppose that $\rho$ and $\sigma$ are states 
on a Hilbert space ${\cal H}$. 
The fidelity between these states can be 
represented as:
\begin{eqnarray*}
F(\rho,\sigma)=\min_{E=\{E_i\}:PVM}
\sum_i \mbox{tr}(\rho E_i)^{1/2} 
\mbox{tr}(\sigma E_i)^{1/2}, 
\end{eqnarray*}
where the minimum 
is taken with respect to all PVMs on ${\cal H}$. 
\end{lemma}
The right-hand side of the above lemma 
represents the degree of overlap between 
the probability distributions 
$\{\mbox{tr}(\rho E_i)\}_i$ and $\{\mbox{tr}(\sigma E_i)\}_i$.
It is $0$ if there is no overlap and $1$ if the
probability distributions coincide with each other. 
Thus, if a process describes an error-free measurement, 
$F(\rho^A_0(\tau),\rho^A_1(\tau))=0$ must be satisfied.
\par While we introduced the fidelity to characterize the 
distinguishability of states, 
in the problem of measurement 
one usually discusses the error probability for a fixed 
observable $Z$ of the apparatus, which is called a pointer observable
or a meter observable. $Z$ has two outcomes $0$ and $1$, and it is 
a projection operator acting on ${\cal H}_A$.    
We define the following quantities for $j=0,1$:
\begin{eqnarray*}
p(1|j)&:=&\mbox{tr}(\rho^A_j(\tau)Z)
\\
p(0|j)&:=&\mbox{tr}(\rho^A_j(\tau)({\bf 1}_A-Z)).
\end{eqnarray*}
That is, $p(i|j)$ represents the conditional probability to 
obtain an outcome $i$ with respect to 
the initial state $|q_j\rangle$ of the system. 
In the error-free case, $p(i|j)=\delta_{ij}$ 
holds. 
In general, none of the $p(i|j)$'s are vanishing. 
 Let us consider an average error defined by 
\begin{eqnarray*}
P_{error}:=\frac{1}{2}(p(1|0) +p(0|1)). 
\end{eqnarray*}
This quantity is related to the fidelity by the 
following lemma.
\begin{lemma}
$P_{error}$ is related to the fidelity as, 
\begin{eqnarray*}
F(\rho^A_0(\tau),\rho^A_1(\tau))\leq 2\sqrt{P_{error}-P_{error}^2}.
\end{eqnarray*}
\end{lemma}
\begin{proof}
Thanks to Lemma \ref{dare}, 
we obtain 
\begin{eqnarray*}
F(\rho^A_0(\tau),\rho^A_1(\tau))^2
&\leq& \left( \sqrt{p(0|0)}\sqrt{p(0|1)}
+\sqrt{p(1|0)}\sqrt{p(1|1)}\right)^2
\\
&=&
p(0|0)p(0|1) +p(1|0)p(1|1)
+2 \sqrt{p(0|0)p(0|1) p(1|0)p(1|1)}
\\
&\leq &
2(p(0|0)p(0|1) +p(1|0)p(1|1))
\\
&=& 
2\left(
(1-p(1|0))p(0|1) +p(1|0)(1-p(0|1))
\right)
\\
&=& 4 \left(
 P_{error}- p(1|0)p(0|1) 
\right) 
\\
&=& 4\left(
P_{error} - p(1|0)(2P_{error} -p(1|0))
\right)
\\
&=& 4\left(
p(1|0)^2 -2P_{error} p(1|0) +P_{error}
\right)\leq 4(P_{error}-P_{error}^2
).
\end{eqnarray*}
\end{proof}
The following is our main theorem. 
\begin{theorem}\label{mainth}
Let us consider a measurement process 
on a two-level system, as introduced above. 
That is, ${\cal H}_S$ denotes the two-dimensional 
Hilbert space of the system, ${\cal H}_A$ denotes the
Hilbert space of the apparatus, and $H=H_S \otimes {\bf 1}_A
+{\bf 1}_S \otimes H_A +V$ denotes the Hamiltonian 
describing the interaction process. 
We consider the measurement process for the observable 
$Q$ that has a pair of eigenstates: 
$Q|q_1\rangle =|q_1\rangle$ and $Q|q_0\rangle =0$. 
For any initial state $|\Omega\rangle \in {\cal H}_A$ 
of the apparatus and time interval $\tau$ for 
the process, 
the following inequality 
holds:
\begin{eqnarray*}
\Vert [Q, H_S]\Vert
\leq \Vert H_S \Vert F(\rho^A_0(\tau),\rho^A_1(\tau))
+\frac{\tau}{\hbar}
\Vert [V,H_S\otimes {\bf 1}_A]\Vert,
\end{eqnarray*} 
where
$\Vert \cdot \Vert$ is the operator norm defined by 
$\Vert A\Vert:=\sup_{|\phi\rangle \neq 0, |\phi\rangle \in {\cal H}}
\frac{\Vert A|\phi\rangle\Vert}{\Vert \phi\Vert}$
for an operator $A$ on ${\cal H}$, and  
 $F(\rho^A_0(\tau),\rho^A_1(\tau))$ represents 
the fidelity between a pair of states on the apparatus after 
the interaction.
\end{theorem}
\begin{proof}
Note that 
time evolution preserves 
the total Hamiltonian. 
That is, $H=U(\tau)^* H U(\tau)$ holds. 
We operate on it 
with $\langle q_0,\Omega |\cdot |q_1,\Omega\rangle$
to obtain, 
\begin{eqnarray}
&& \langle q_0 | H_S |q_1\rangle
+\langle q_0,\Omega |{\bf 1}_S \otimes H_A |q_1,\Omega \rangle 
+\langle q_0,\Omega |V |q_1,\Omega \rangle 
\nonumber \\
&& = \langle q_0,\Omega|U(\tau)^* (H_S\otimes {\bf 1}_A)
 U(\tau)|q_1,\Omega \rangle 
+\langle q_0,\Omega |U(\tau)^*({\bf 1}_S \otimes  H_A)
 U(\tau) |q_1,\Omega \rangle 
+\langle q_0,\Omega |U(\tau)^* V U(\tau)|q_1,\Omega \rangle, 
\label{saisho}
\end{eqnarray}
where we keep the second term of the left-hand side 
although  
$\langle q_0,\Omega |{\bf 1}_S \otimes H_A |q_1,\Omega \rangle
=0$ holds owing to the orthogonality 
of $|q_0\rangle$ and $|q_1\rangle$. 
Equation (\ref{saisho}) can be further reduced as
\begin{eqnarray}
|\langle q_0 | H_S |q_1 \rangle |
\leq |\langle q_0,\Omega |U(\tau)^* (H_S\otimes 
{\bf 1}_A) U(\tau) |q_1,\Omega \rangle |
+ |\langle q_0,\Omega |
U(\tau)^*({\bf 1}_S\otimes H_A 
+V)U(\tau) -({\bf 1}_S \otimes H_A+V)| q_1,\Omega \rangle |, 
\label{tugi}
\end{eqnarray}
where the triangular inequality was used. 
The first term of the right-hand side can be 
bounded as follows. Let us consider an arbitrary 
PVM $E=\{E_i\}_i$ on ${\cal H}_A$. 
Because $\sum_i E_i={\bf 1}_A$ holds, it follows that
\begin{eqnarray*}
|\langle q_0,\Omega |U(\tau)^* (H_S\otimes {\bf 1}_A)
 U(\tau) |q_1,\Omega \rangle |
=|\sum_i \langle q_0,\Omega |U(\tau)^*({\bf 1}_S
\otimes E_i)( H_S \otimes {\bf 1}_A) 
U(\tau) |q_1,\Omega \rangle |.
\end{eqnarray*}
The commutativity between ${\bf 1}_S \otimes E_i$ and $H_S
\otimes {\bf 1}_A$ allows the further derivation 
\begin{eqnarray*}
&&|\sum_i \langle q_0,\Omega |U(\tau)^*({\bf 1}_S\otimes  E_i)
( H_S \otimes {\bf 1}_A) U(\tau) 
|q_1,\Omega \rangle |\\
&=&|\sum_i q_0,\Omega |U(\tau)^*({\bf 1}_S \otimes  E_i) 
(H_S\otimes {\bf 1}_A)( {\bf 1}_S \otimes 
E_i) U(\tau) |q_1,\Omega \rangle |
\\
&\leq &
\sum_i |\langle q_0,\Omega |U(\tau)^* ({\bf 1}_S \otimes E_i)
( H_S \otimes {\bf 1}_A) ({\bf 1}_S \otimes E_i) 
U(\tau) |q_1,\Omega \rangle | \\
&\leq &
\sum_i \Vert H_S \Vert 
\langle q_0, \Omega |U(\tau)^* 
({\bf 1}_S \otimes E_i) U(\tau) |q_0,\Omega\rangle^{1/2}
\langle q_1,\Omega |U(\tau)^* ({\bf 1}_S \otimes E_i)
 U(\tau)|q_1,\Omega\rangle^{1/2}
\\
&=& 
\Vert H_S\Vert 
\sum_i \mbox{tr}(\rho^A_0(\tau) E_i)^{1/2}
\mbox{tr}(\rho^A_1(\tau)E_i)^{1/2}, 
\end{eqnarray*}
where the Cauchy-Schwarz inequality was used. 
Because the choice of a PVM $\{E_i\}$ is arbitrary, 
applying Lemma \ref{dare} we obtain 
\begin{eqnarray*}
|\langle q_0,\Omega |U(\tau)^* (H_S\otimes {\bf 1}_A)
 U(\tau) |q_1,\Omega \rangle |
\leq \Vert H_S \Vert F(\rho^A_0(\tau),\rho^A_1(\tau)).
\end{eqnarray*}
The second term of (\ref{tugi}) can be 
bounded as follows. 
Applying the conservation of the total Hamiltonian we obtain 
\begin{eqnarray*}
U(\tau)^*({\bf 1}_S \otimes H_A +V)U(\tau) -
({\bf 1}_S \otimes H_A+V)
=H_S\otimes {\bf 1}_A -U(\tau)^* (H_S\otimes {\bf 1}_A) U(\tau). 
\end{eqnarray*}
Its right-hand side is bounded by using 
the Heisenberg equation. 
Because $U(t)^* (H_S \otimes {\bf 1}_A) U(t)$ satisfies 
\begin{eqnarray*}
i\hbar \frac{d}{dt}U(t)^* (H_S\otimes {\bf 1}_A) U(t)
=U(t)^*[H_S\otimes {\bf 1}_A, H]U(t),
\end{eqnarray*}
we obtain
\begin{eqnarray*}
U(\tau)^* (H_S\otimes {\bf 1}_A) U(\tau) -H_S\otimes {\bf 1}_A 
= \frac{1}{i\hbar} 
\int^{\tau}_0 dt U(t)^* [H_S\otimes {\bf 1}_A, V] U(t), 
\end{eqnarray*}
which derives 
\begin{eqnarray*}
|\langle q_0,\Omega |U(\tau)^* 
({\bf 1}_S\otimes H_A+V)U(\tau)-({\bf 1}_S\otimes H_A +V)|q_1,\Omega\rangle |
&=&
|\langle q_0,\Omega |
H_S\otimes {\bf 1}_A -U(\tau)^* (H_S \otimes {\bf 1}_A)
U(\tau)|q_1,\Omega\rangle |
\\
&\leq&
\Vert H_S \otimes {\bf 1}_A -
U(\tau)^* (H_S \otimes {\bf 1}_A) U(\tau) \Vert 
\\
&\leq & 
\frac{1}{\hbar}\int^{\tau}_0 dt \Vert 
U(t)^* [H_S \otimes {\bf 1}_A, V]U(t)\Vert 
\\
&=& 
\frac{\tau}{\hbar} \Vert [H_S\otimes {\bf 1}_A, V]\Vert .
\end{eqnarray*}
Finally, we analyze the left-hand side of Equation (\ref{tugi}). 
We derive the equality $|\langle q_0 |H_S|q_1\rangle |
=\Vert [Q,H_S]\Vert$ in the following. 
Because $i[Q,H_S]$ is a self-adjoint operator, 
$\Vert i[Q,H_S]\Vert = \max_{|\psi\rangle:\Vert |\psi\rangle
\Vert =1}
| \langle \psi |i[Q,H_S] |\psi\rangle |$ holds. 
For any unit vector $|\psi\rangle =
c_0 |q_0\rangle +c_1 |q_1 \rangle$, we have 
\begin{eqnarray*}
| \langle \psi |i[Q,H_S]|\psi \rangle |
=2 | \mbox{Im} (\overline{c_0}c_1 \langle q_0 |H_S |q_1\rangle ) |.  
\end{eqnarray*}
Its right-hand side can be bounded as 
\begin{eqnarray*}
2 | \mbox{Im}(
 \overline{c_0}c_1 \langle q_0 |H_S |q_1\rangle )|
&\leq& 2 | \overline{c_0}| |c_1 | |\langle q_0 |H_S |q_1\rangle| \\
&\leq &(|c_0|^2 +|c_1 |^2) |\langle q_0 |H_S |q_1\rangle | \\
&=&  |\langle q_0 |H_S |q_1\rangle | . 
\end{eqnarray*}
In the above inequality, 
if we insert $c_0 =\frac{\langle q_0 |H_S |q_1 \rangle}
{\sqrt{2}|\langle q_0 |H_S |q_1 \rangle |}$
and $c_1 =\frac{1}{\sqrt{2}}$, 
we obtain an equality. 
Thus, we proved 
\begin{eqnarray*}
\Vert [Q, H_S]\Vert
=\Vert i[Q,H_S]\Vert 
=|\langle q_0 |H_S |q_1\rangle |.
\end{eqnarray*}
It ends the proof.
\end{proof}
These corollaries immediately follow: 
\begin{corollary}
Based on the above theorem, the following holds for any 
pointer observable $Z$: 
\begin{eqnarray*}
\Vert [Q,H_S] \Vert \leq 
2 \Vert H_S \Vert \sqrt{P_{error} -P_{error}^2}+
\frac{\tau}{\hbar} \Vert 
[V, H_S\otimes {\bf 1}_A]\Vert. 
\end{eqnarray*}
\end{corollary}
\begin{corollary}\label{cor2}
In order to attain an error-free measurement, 
the interaction Hamiltonian $V$ and 
the time interval $\tau$ must satisfy
\begin{eqnarray*}
\tau \cdot \Vert 
[V, H_S\otimes {\bf 1}_A]\Vert
\geq \hbar \Vert [Q,H_S]\Vert
. 
\end{eqnarray*}
\end{corollary}
The above theorem and corollaries show that if a measured observable 
does not commute with the system Hamiltonian, 
then there exists a non-trivial 
trade-off relationship among the strength of interaction, 
the time interval, and the accuracy of the measurement. 
Particularly, according to Corollary \ref{cor2}, 
in order to achieve an error-free measurement process for such an observable, 
the interaction Hamiltonian must be noncommutative with 
the system Hamiltonian.
Note that the inequalities do not 
contain the Hamiltonian of the apparatus. 
In the discussion, we give a brief interpretation 
of this result. 
\subsection{When the observable 
commutes with the system Hamiltonian}
\label{subsec:commute}
When the observable $Q$ commutes with $H_S$, 
the above theorem  becomes trivial. 
The following example shows that in such a case 
$V$ can commute with $H_S$ and $\tau$ can be arbitrarily 
small even for an error-free measurement.
\par 
Let us consider the standard model of the measurement process 
\cite{BuschStandard}. 
Assume that the system is a two-level system 
described by the two-dimensional Hilbert space 
${\cal H}_S$, and assume that the apparatus is 
a particle moving in one degree of freedom 
so that ${\cal H}_A=L^2({\bf R})$. 
The observable to be measured is denoted by $Q$, 
which is a projection operator with eigenstates 
$|q_1\rangle$ and $|q_0\rangle$. 
Assume that the total Hamiltonian is defined by 
\begin{eqnarray*}
H&=&H_S + H_A +V \\
&=& {\bf 0} +{\bf 0}
+ Q \otimes P_{\cal A}, 
\end{eqnarray*}
where $P_{\cal A}$ is the momentum operator 
of the apparatus. Because $H_S$ is trivial, 
both $[Q,H_S]=0$ and $[V,H_S\otimes {\bf 1}_A]=0$ hold. 
The initial states $|q_0\rangle \otimes |\Omega \rangle$ 
and $|q_1\rangle \otimes |\Omega\rangle$
for any $|\Omega\rangle \in {\cal H}_A$ evolve as
\begin{eqnarray*}
|q_0\rangle \otimes |\Omega\rangle
&\to& |q_0\rangle \otimes |\Omega_0 \rangle ,
\\
|q_1\rangle\otimes |\Omega\rangle 
&\to& |q_1\rangle \otimes |\Omega_{\frac{\tau}{\hbar}}\rangle,
\end{eqnarray*}
where $|\Omega_{\lambda}\rangle$ is defined by 
$\langle x|\Omega_{\lambda}\rangle 
:=\langle x-\lambda |\Omega\rangle$ for any $\lambda \in {\bf R}$
in the position representation. 
Therefore, for any $\tau>0$, 
if we prepare $|\Omega\rangle$ so that 
it is sharply localized in the position representation, 
$|\Omega_{\frac{\tau}{\hbar}}\rangle$ and $|\Omega_0\rangle$ 
become disjoint in the position representation, and are 
perfectly distinguishable. 
\par
While the interaction Hamiltonian in the above example 
is unbounded in norm, the following example shows that 
the norm of the interaction Hamiltonian can be made arbitrarily small. 
We consider a two-level system interacting with 
a two-level apparatus,  
whose dynamics is governed by 
$H=V= \lambda (|q_1 \rangle \langle q_1 |
\otimes |1\rangle \langle 1|
+ |q_0\rangle \langle q_0| \otimes |0\rangle \langle 0|)$ for $\lambda >0$, 
where $|q_0\rangle $ and $|q_1\rangle $ are eigenstates of $Q$. 
If we put the initial state of the apparatus as 
$|\Omega\rangle :=\frac{1}{\sqrt{2}}(|0\rangle + |1\rangle)$, 
the initial states $|q_1\rangle \otimes |\Omega \rangle $ and 
$|q_0\rangle \otimes |\Omega \rangle $ evolve as
\begin{eqnarray*}
|q_1 \rangle \otimes |\Omega\rangle 
&\to & 
|q_1\rangle \otimes \frac{1}{\sqrt{2}}(|0\rangle - i |1\rangle)
\\
|q_0\rangle \otimes |\Omega\rangle &\to& 
|q_0\rangle \otimes \frac{1}{\sqrt{2}}(|0\rangle +i |1\rangle)
\end{eqnarray*} 
in time $\tau:=\frac{\pi \hbar}{2\lambda}$. 
Note that $\Vert V\Vert =\lambda$ can be arbitrarily small. 
\section{Discussion}
In this paper, we applied Hamiltonian formalism to 
 the measurement process 
of a two-level system. 
For a measured observable that does not commute with 
the system Hamiltonian, 
 a non-trivial trade-off 
 relationship among the strength of interaction, 
the time interval, and the accuracy of the measurement has been obtained. 
In particular, 
in order to achieve an error-free measurement process for such an observable, 
the interaction Hamiltonian must be noncommutative with 
the system Hamiltonian.
We show that this impossibility result can be derived 
by the uncertainty principle for joint measurement. 
Let us consider a two-level quantum system ${\cal H}_S$
interacting with an apparatus ${\cal H}_A$. 
We denote the total Hamiltonian by $H=
H_S\otimes {\bf 1}_A+{\bf 1}_S\otimes H_A +V$. 
The system Hamiltonian $H_S$ is diagonalized as 
$H_S=\epsilon_1 |\epsilon_1\rangle\langle \epsilon_1|
+\epsilon_0 |\epsilon_0\rangle \langle \epsilon_0|$ 
($\epsilon_0\neq \epsilon_1$). 
The unitary evolution $U(\tau) =
\exp(-i \frac{H \tau}{\hbar})$ and the initial state 
of the apparatus 
$|\Omega\rangle$ define an 
isometry $W:{\cal H}_S \to {\cal H}_S \otimes {\cal H}_A$ by 
$W|\psi\rangle :=U(\tau)|\psi\rangle \otimes |\Omega\rangle$. 
If the process describes an error-free measurement of $Q$, 
there exists a PVM $M=\{M_0,M_1\}$ on ${\cal H}_A$ 
satisfying $|q_j\rangle\langle q_j|
=W^*({\bf 1}_S\otimes M_j)W$ for $j=0,1$. 
In addition, if $V$ commutes with $H_S$, 
one obtains $|\epsilon_n\rangle 
\langle \epsilon_n|=W^*(|\epsilon_n\rangle 
\langle \epsilon_n| \otimes {\bf 1}_A)W$ for $n=0,1$. 
Thus we can introduce a positive-operator-valued measure $Y=\{Y_{nj}\}$ by 
$Y_{nj}:=W^*(|\epsilon_n\rangle \langle \epsilon_n|\otimes M_j)W$ 
which jointly measures $Q$ and $H_S$. 
That is, $Y_{n0}+Y_{n1}=|\epsilon_n\rangle \langle \epsilon_n|$ 
and $Y_{0j}+Y_{1j}=|q_j\rangle \langle q_j|$ hold for $n,j=0,1$.   
According to the uncertainty principle for joint measurement, 
these relations can be true only for commutative PVMs $\{|\epsilon_0\rangle
\langle \epsilon_0|, |\epsilon_1\rangle \langle \epsilon_1 |\}$ and
$\{|q_0\rangle \langle q_0 |, |q_1\rangle \langle q_1 |\}$ 
(see for e.g. \cite{Buschbook}). 
\par
As an example showing that an interaction Hamiltonian 
satisfying $[V, H_{S}]\neq 0$ helps reducing error, 
we consider a modified version of the standard model. 
Setting $|\Omega\rangle$ sharply located, 
we denote by $\tau >0$ the time interval required to 
accomplish error-free measurement in the standard model 
discussed in Sec. \ref{subsec:commute}. 
Let us consider a modified model described by 
$H=H_S\otimes {\bf 1} + Q\otimes P_{A}$, 
where $[H_S,Q]\neq 0$ holds. 
Because the time evolution $U(\tau)
=\exp(-i\frac{H \tau}{\hbar})$ satisfies 
$\Vert U(\tau) -\exp(-i \frac{Q\otimes P_A}{\hbar}\tau)\Vert 
\leq \frac{\tau}{\hbar}\Vert H_S\Vert$, 
we obtain an estimate $F(\rho^A_0(\tau), \rho^A_1(\tau))\leq 
2\sqrt{\frac{2\tau}{\hbar}\Vert H_S \Vert}$. 
Note that it is possible to make $\tau$ arbitrarily small by making 
the initial state $|\Omega\rangle$ sufficiently sharp. 
\par
Similarly, it is possible to treat a modified version of the second example 
in Sec. \ref{subsec:commute}. 
If we take $\lambda >0$ sufficiently large 
for $H=H_S\otimes {\bf 1}_A+ V= 
H_S\otimes {\bf 1}_A+ \lambda (|q_1 \rangle \langle q_1 |
\otimes |1\rangle \langle 1|
+ |q_0\rangle \langle q_0| \otimes |0\rangle \langle 0|)$, 
$\tau =\frac{\pi \hbar}{2\lambda}$ becomes small and the fidelity 
between final states on the apparatus can be made arbitrarily small. 
%
\par
While the obtained inequality is non-trivial, 
it may not always be strong. 
In fact, although in the proof of theorem \ref{mainth}, 
we have used the inequality
\begin{eqnarray*}
|\langle q_0,\Omega |H_S\otimes {\bf 1}_A-U(\tau)^* 
(H_S\otimes {\bf 1}_A) U(\tau)|q_1,\Omega\rangle |
&\leq&
\Vert H_S\otimes {\bf 1}_A -U(\tau)^* (H_S\otimes {\bf 1}_A) U(\tau) \Vert 
\\
&\leq & 
\frac{\tau}{\hbar} \Vert [H_S\otimes {\bf 1}_A, V]\Vert, 
\end{eqnarray*}
which was the origin of the linear 
term with respect to $\tau$, 
this bound is not strong because it does not take into 
consideration the dynamics in detail. 
The left-hand side of the above inequality 
can be written as 
\begin{eqnarray*}
|\langle q_0,\Omega |H_S\otimes {\bf 1}_A -U(\tau)^* 
(H_S\otimes {\bf 1}_A)U(\tau)|q_1,\Omega\rangle |
&\leq &
\frac{1}{\hbar}\int^{\tau}_0 dt 
\left|
\langle q_0,\Omega |
 U(t)^* [H_S\otimes {\bf 1}_A,V]U(t)
|q_1,\Omega\rangle
\right|. 
\end{eqnarray*}
The last term contains 
the correlation function $\langle q_0,\Omega |
 U(t)^* [H_S\otimes {\bf 1}_A,V]U(t)
|q_1,\Omega\rangle$ that, in general, decays rapidly 
for a large apparatus. Therefore, the 
term may not grow proportionally to 
$\tau$ in physically realistic models. 
We hope to investigate this problem in the future. 
\\
{\bf Acknowledgments:} 
I would like to thank an anonymous referee for helpful comments. 
%

\end{document}